\documentclass[lettersize,journal]{IEEEtran}
\usepackage{cite}
\usepackage{amsmath,amssymb,amsfonts, amsthm}
\usepackage{algorithmic}
\usepackage{graphicx}
\usepackage{textcomp}
\usepackage{subcaption}
\usepackage{wrapfig}
\usepackage[ruled,vlined]{algorithm2e}
\usepackage{array}
\usepackage[caption=false,font=normalsize,labelfont=sf,textfont=sf]{subfig}
\usepackage{stfloats}
\usepackage{url}
\usepackage{verbatim}
\usepackage{color,soul}
\newtheorem{prop}{Proposition}

\begin{document}

\title{Two-Stage Distributed Beamforming Design in Cell-Free Massive MIMO ISAC Systems}

\author{Leonardo Leyva,~\IEEEmembership{Graduate Student Member,~IEEE,} Daniel Castanheira, Adão Silva,~\IEEEmembership{Member,~IEEE,} and Atílio Gameiro 
\thanks{Manuscript received April 19, 2021; revised August 16, 2021.}
\thanks{Leonardo Leyva, Daniel Castanheira, Adão Silva and Atílio Gameiro are with the Instituto de Telecomunicações, and with Departamento de Electrónica, Telecomunicações e Informática, Universidade de Aveiro, 3810-164, Aveiro, Portugal. (E-mails: leoleval@av.it.pt, dcastanheira@av.it.pt, asilva@av.it.pt, amg@ua.pt).}}

\markboth{Journal of \LaTeX\ Class Files,~Vol.~14, No.~8, August~2021}%
{Shell \MakeLowercase{\textit{et al.}}: A Sample Article Using IEEEtran.cls for IEEE Journals}


\maketitle

\begin{abstract}
Integrating radio-sensing functionalities into future cell-free (CF) wireless networks promises efficient resource utilization and facilitates the seamless roll-out of applications such as public safety and smart infrastructure. While the beamforming design problem for the CF integrated sensing and communication (ISAC) paradigm has been addressed in the literature, existing methods rely on centralized signal processing, leading to fronthaul load and scalability issues.   
This paper presents a two-stage beamforming design for the CF ISAC paradigm, aiming to significantly reduce the fronthaul load by distributing the signal processing tasks between the central unit (CU) and the access points (APs). The design optimizes the sum signal-to-interference-plus-noise ratio (SINR) for communication users, subject to per-AP power constraints and signal-to-noise ratio (SNR) requirements for radio-sensing purposes. The resulting optimization problems are non-convex and challenging to solve. To address this, we employ a majorization-minimization (MM) approach, which decomposes the problem into simpler convex subproblems. 
The results show that the two-stage beamforming design achieves performance comparable to centralized methods while substantially reducing the fronthaul load, thus minimizing data transmission requirements over the fronthaul network. This advancement highlights the potential of the proposed method to enhance the efficiency and scalability of cell-free MIMO ISAC systems. 
\end{abstract}


\section{Introduction}
\IEEEPARstart{T}{he} integration of radio-sensing and communication functionalities, an innovative paradigm known as integrated sensing and communication (ISAC), stands as a key component in the landscape of next-generation wireless networks. ISAC promises to revolutionize applications by enabling simultaneous data transmission and environmental sensing. 

\subsection{Background} 

ISAC represents a significant shift from traditional wireless systems, where radio-sensing and communication functions are handled by separate and dedicated systems. By tightly integrating these functions, ISAC promises more efficient use of the radio spectrum and reduces hardware and infrastructure resources  \cite{liu2020, paul2017, leyva2021}. This integration is not only advantageous in terms of cost and resource savings but also enables new use cases in several areas, such as intelligent transportation systems (ITS), autonomous driving, smart home, weather monitoring, and smart cities \cite{zhang2022, cui2021, leyva2021}. For example, in autonomous vehicles, real-time traffic sensing and communication between cars and roadside units (RSUs) are essential. Smart cities necessitate efficient environmental monitoring and data exchange for effective urban infrastructure management.   

The potential of ISAC has driven growing interest in integrating radio-sensing functions into future wireless networks. The last decade has witnessed much academic research on ISAC \cite{zhang2022, cui2021, liu2023}. More recently, standards development organizations like 3GPP and ETSI have initiated pre-standardization activities to support this integration \cite{3gpp2023ts, 3gpp2023tr}. In the context of 6G, ISAC technologies are essential for enabling precise localization and tracking of connected and unconnected objects, environment reconstruction, and enhanced situational awareness \cite{leyva2021, liu2022isc, giordani2020toward}. These capabilities are vital for advanced applications such as public safety, immersive experiences, and smart infrastructure. While the possibilities presented by ISAC in 6G are promising, they also introduce significant research challenges in areas such as waveform design, including beamforming optimization, and networked sensing for wireless networks \cite{zhang2022}.

In the ISAC context, beamforming involves designing precoding matrices that jointly and efficiently support both communication and radio-sensing functions. This area has seen substantial research in recent years, with contributions to the design of digital and fully/partially-connected hybrid beamforming \cite{liu18, eldar20, leyva2024hybrid}.
For example, in \cite{liu18}, the authors design precoding matrices such that the transmitted beampattern approximates an optimal radio-sensing probing beam, while also meeting SINR constraints for communication users. They assume that the communication signals are used as the radio-sensing probing beam. The results demonstrated that the ISAC transmission yields better overall performance than the coexistence scenario, where antennas are separated to deliver communication and radio-sensing operations separately.
The authors of \cite{eldar20} also proposed a radio-sensing centric method, i.e., minimizes a radio-sensing loss function while satisfying the SINR of communication users and power budget constraints. Differently, to \cite{liu18}, the authors assumed different communication and radio-sensing signals. The proposed method increases the degrees of freedom (DoF) for the design of the radio-sensing waveform but introduces interference at the communication receivers.
More recently, \cite{leyva2024hybrid} proposed the use of hybrid beamforming instead of fully-digital beamforming in a multi-user, multi-beam ISAC scenario, where communication signals are also employed for radio-sensing. This iterative approach leverages the fully-connected hybrid architecture to efficiently integrate communication and radio-sensing functionalities while significantly reducing hardware costs and energy consumption.
Previous contributions concentrate on a single base station (BS), which has inherent limitations compared to distributed networks. More specifically, using multiple BSs can provide superior communication coverage and more accurate radio-sensing performance by leveraging channel diversity.

\subsection{Related Work}

Distributed networks have recently garnered attention for ISAC systems. The advantages of a distributed networked ISAC include enhanced coverage, increased diversity, a reduced likelihood of object shadowing, and greater robustness to environmental changes. These benefits have prompted several contributions in the area of optimal precoding design for multiple-BS ISAC approaches \cite{Li2022, Chen2023, Cheng2024}.

In \cite{Li2022}, the authors considered a two-cell interfering ISAC scenario. They proposed optimal/sub-optimal solutions for precoders and filters that minimize the transmit power while satisfying SINR constraints for communication users and targets. Although the results show benefits compared to benchmark schemes, the proposed method is limited to two cells and single-target detection per BS.    
Joint transmission (JT) coordinated multi-point (CoMP) ISAC system, where $L$ BSs are coordinated by a central processing unit (CPU). The proposed algorithm optimizes the waveform and clustering designs, from radar- and communication-centric perspectives. The results showcased better performance for optimal clustering design compared to a static design benchmark.   
Similarly, the authors of \cite{Cheng2024} considered a CoMP ISAC system. Unlike \cite{Chen2023}, which considers a JT framework, they assumed a coordinated beamforming (CB) approach, where each BS only serves its corresponding user equipments (UEs). It is considered two joint detection scenarios with and without time synchronization among the BSs.  The results demonstrated the benefits of time synchronization among BSs for improved joint detection and communication performance. The results also showed a higher detection probability than other benchmark schemes. However, the CB CoMP may cause severe intercell interference especially when the number of UEs becomes large. 

Cell-free (CF) wireless communication systems offer advantages over classical cellular architectures by eliminating inter-cell interference and providing uniform quality of service (QoS) across the coverage area, thus enhancing the user experience \cite{ngo2017cellfree}. The distributed nature of CF systems leverages spatial diversity to improve radio-sensing performance and increase environmental awareness, making it a promising approach for enabling ISAC capabilities. Recently, there has been increased interest in the CF-ISAC paradigm, leading to several contributions \cite{Behdad2022, Behdad2023, Cao2023, Demirhan2023, Mao2023, Mao2023GlobeCom, Cao2023CF, Liu2024}.

The authors of \cite{Behdad2022} considered a centralized CF paradigm for an ISAC scenario, where the Access Points (APs) jointly perform downlink communications and multi-static sensing of a single target. They proposed power allocation strategies to maximize the sSNR while satisfying communication constraints. The considered algorithm is developed under the condition that the target is present and no clutter.
More recently, \cite{Behdad2023} extended the work initiated in \cite{Behdad2022} by considering target-free scenarios and the presence of clutter. However, the algorithms in \cite{Behdad2022,Behdad2023} assume a unit-norm regularized zero-forcing (RZF) precoding vector, which reduces the system performance. 
Additionally, \cite{Cao2023} investigated a user-centric CF ISAC system, where, besides the power allocation scheme, the authors also proposed an efficient scheduling method to optimize the sum-rate performance. However, they adopted conjugate beamforming for the precoder, which further reduces system performance. Although power allocation and scheduling are interesting for the CF ISAC paradigm, the beamforming design brings additional DoF that could enhance system performance.

Spatial optimization through beamforming design is a crucial aspect of CF ISAC systems, and its importance has grown significantly in recent years \cite{Demirhan2023, Mao2023GlobeCom, Mao2023, Cao2023CF, Liu2024}.  
In \cite{Demirhan2023}, the authors designed the transmit beamforming in a CF ISAC scenario, where the transmitter APs jointly handle downlink communication with multiple users and steer a radio-sensing probing beam toward a target. They focused on maximizing the sSNR for multi-static sensing while subject to SINR and per-AP power budget constraints. The proposed design outperforms communication-prioritized and sensing-prioritized baseline methods used as benchmarks.
The authors of \cite{Mao2023GlobeCom} considered a scenario similar to that in \cite{Demirhan2023}, but they derived a solution by minimizing the mean square error in the sensing beampattern matching problem, subject to power budget constraints and the ergodic rate requirements of communication users. They assumed imperfect CSI and that each target is assigned to a single AP.  
In \cite{Mao2023}, the work was extended to derive the communication-sensing (C-S) region, which evaluates the trade-off between radio-sensing-only, communication-only, and joint beamforming approaches.  
More recently, \cite{Liu2024} addressed the problem of joint beamforming/filter design and transmitter-receiver AP mode selection. The APs simultaneously handle downlink communications and detect several point-like targets. The authors jointly derived the transmit beamforming, receiver filters, and AP mode configuration that maximizes the sum of the sensing SNR, subject to power budget constraints, SINR requirements, and BS mode selection constraints. The results highlighted the importance of BS mode selection in a CF ISAC paradigm.
The authors of \cite{Cao2023CF} proposed a communication-centric approach, where the maximization of the communication rate is prioritized, with radar estimation rate and power budget constraints considered. This approach is particularly advantageous for applications prioritizing communication performance, ensuring high data rates are maintained even when radio-sensing functionalities are required.

In \cite{Demirhan2023, Mao2023GlobeCom, Mao2023, Cao2023CF, Liu2024}, signal processing is conducted at the CU. This centralized approach leads to high computational complexity and increased fronthaul load, particularly due to the transmission of the CSI. As a result, these solutions exhibit scalability challenges when the number of APs, UEs, and antenna array dimensions increase in CF ISAC scenarios.

\subsection{Contributions}

This work addresses the limitations of existing approaches by proposing a novel algorithm that distributes signal-processing tasks between the CU and the APs. This method reduces computational complexity at the CU and alleviates the fronthaul load during beamforming design, enabling a more scalable CF ISAC system. The key contributions of this work are as follows:

\begin{enumerate}
	\item Novel Two-stage Distributed Signal Processing Approach for Beamforming Design in CF-ISAC: We introduce a two-stage distributed signal processing framework where the beamforming design tasks are divided between the CU and APs. This division reduces the processing burden on the CU.
	\item Compact Information Exchange to Reduce Fronthaul Load: To reduce the high data transfer usually required between the APs and the CU in CF scenarios, we propose transmitting compact signal information, such as equivalent channels and powers, between the APs and the CU. This approach significantly reduces the fronthaul load by minimizing the amount of data that needs to be transmitted, without compromising the performance of the system.
	\item Iterative Algorithms Based on Majorization-Minimization (MM): We develop iterative algorithms using the Majorization-Minimization (MM) technique to solve the optimization problems at both the CU and the APs. 
\end{enumerate}

The results demonstrate similar communication performance while reducing computational complexity and fronthaul load compared to the centralized-based solutions.

\subsection{Paper Structure and Notation}

The structure of the paper is as follows. Section \ref{sectII_CF} introduces the CF ISAC system model, covering the transmitted signal, communication receiver, and radio-sensing models. Section \ref{distBemFra} introduces a distributed beamforming framework, where the precoders are partitioned between the APs and CU. Section \ref{sectIII_CF} details the proposed beamforming design method, including the problem formulation, iterative two-stage distributed approach, and the algorithms for solving the optimization problems at the APs and CU. Section \ref{sectIV_CF} evaluates and compares the proposed algorithm to the centralized solution. Finally, Section \ref{sectV_CF} summarizes the main conclusions of the paper.

\textit{Notation}: 	Complex scalars are represented by normal font, i.e., $a$, vectors and matrices are denoted by bold lowercase and bold uppercase letters, respectively, i.e., $\mathbf{a}$, and $\mathbf{A}$.  $\mathbb{E}$ stands as the expectation operator, $\mathbb{C}$ and $\mathbb{R}$ denote the set of complex and real numbers, respectively. Besides, $[\,  \cdot \,]^T$ , $[ \, \cdot \,]^*$, and $[\,  \cdot \,]^H$, indicate the transpose, conjugate and Hermitian transpose operations, respectively. Finally, $\mathrm{diag}(\cdot)$ denotes the diagonal matrix of a vector, $\Re(\cdot)$ represents the real part of a complex scalar, and $||\cdot||_F$ denotes the Frobenius norm.

\section{System Model} \label{sectII_CF}

This section outlines the system model considered in this work. Initially, it introduces the cell-free (CF) integrated sensing and communications (ISAC) scenario. Then, the transmitted signal model is presented. Subsequently, the communication and radio-sensing models are described in detail.      

\subsection{Cell-Free ISAC Scenario}

Fig. \ref{00_fig_1_multistatic_CF_isac_scenario} illustrates the considered CF ISAC scenario where the APs are divided into $M$ transmitters ($\mathrm{AP}_{\mathrm{tx}}$) and $N$ receivers ($\mathrm{AP}_{\mathrm{rx}}$). These APs are connected to the CU via fronthaul links. The set of $\mathrm{AP}_{\mathrm{tx}}$/$\mathrm{AP}_{\mathrm{rx}}$ forms a multistatic topology within the CF ISAC scenario, that enables radio-sensing applications \cite{leyva2021}. Also, $\mathrm{AP}_{\mathrm{tx}}$ and $\mathrm{AP}_{\mathrm{rx}}$ hold uniform linear arrays (ULAs) with $N_{\mathrm{tx}}$ and $N_{\mathrm{rx}}$ antenna elements, respectively.   

\begin{figure}[!htb]
	\centering
	\includegraphics[width=0.49\textwidth]{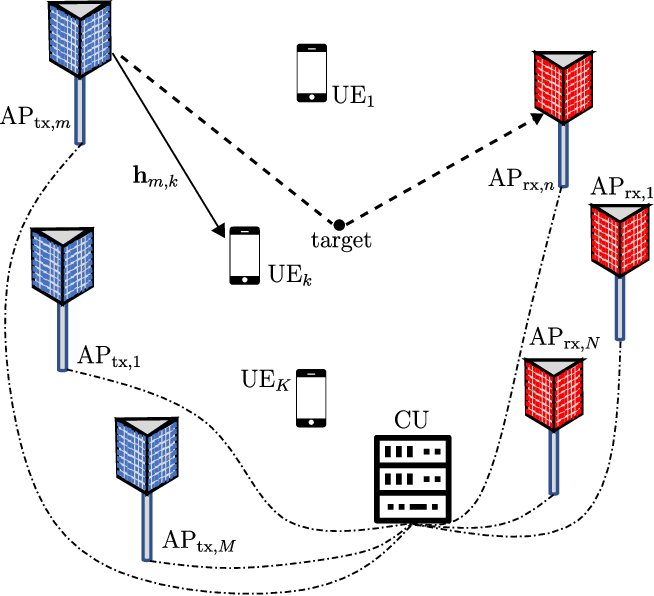}
	\caption{Illustration of the considered multistatic CF mMIMO ISAC scenario.}
	\label{00_fig_1_multistatic_CF_isac_scenario}
\end{figure} 

This work considers a scenario where the set of $\mathrm{AP}_{\mathrm{tx}}$ communicates in the downlink with $K$ single-antenna UEs, while simultaneously steering a single probing beam in the desired direction for target detection and estimation.  Additionally, the set of $\mathrm{AP}_{\mathrm{rx}}$ is configured to function exclusively as radio-sensing receivers. Specifically, the signals received by the $\mathrm{AP}_{\mathrm{rx}}$ are processed to detect and estimate objects in the environment.

We make the following assumptions about the CF ISAC scenario: 1) each UE is connected to all \(\mathrm{AP}_{\mathrm{tx}}\), and 2) all \(\mathrm{AP}_{\mathrm{tx}}\) steer their radio-sensing beam towards the same geographic location (a point-like target, as shown in Fig. \ref{00_fig_1_multistatic_CF_isac_scenario}). Additionally, we consider that all APs are synchronized and that the channel state information (CSI) between the \(\mathrm{AP}_{\mathrm{tx}}\) and the UEs is perfectly known by the \(\mathrm{AP}_{\mathrm{tx}}\). We emphasize that each \(\mathrm{AP}_{\mathrm{tx}}\) has knowledge only of its local channel to the UEs, and not the channels of other \(\mathrm{AP}_{\mathrm{tx}}\) to the UEs.
 
Regarding the radio-sensing function, the set of $\mathrm{AP}_{\mathrm{tx}}$/$\mathrm{AP}_{\mathrm{rx}}$ cooperatively scan the environment under the orchestration of the CU. For this, the CU calculates the angles corresponding to a specific position where an alleged target might be located and forwards this information to the APs, allowing them to direct the radio-sensing beams toward that position. Additionally, both the $\mathrm{AP}_{\mathrm{tx}}$ and $\mathrm{AP}_{\mathrm{rx}}$ have a line of sight (LoS) to the target, though a LoS path between the $\mathrm{AP}_{\mathrm{tx}}$ and $\mathrm{AP}_{\mathrm{rx}}$ is not required.

\subsection{Transmitted Signal}
 
This subsection presents the joint communication and radio-sensing signal model used in the downlink transmissions. More specifically, each $\mathrm{AP}_{\mathrm{tx}}$ transmits $K$ data streams towards the UEs, where the communication data streams are also used for radio-sensing purposes \cite{liu18,liu2020,leyva2024hybrid}. The signal transmitted by the $m$th $\mathrm{AP}_{\mathrm{tx}}$ is expressed as
\begin{equation} 
	\label{II_eq01_tx_signal}
	\begin{aligned}
		\mathbf{x}_m & 	= \mathbf{F}_m \mathbf{s}\\
		&	= \sum_{k=1}^K \mathbf{f}_{m,k}s_k
	\end{aligned}
\end{equation}
where $\mathbf{f}_{m,k} \in \mathbb{C}^{N_{\mathrm{tx}}}$ represents the beamforming vector corresponding to the $k$th data stream, and $\mathbf{s} = [s_1, \cdots ,s_K]^T$ is the vector of the communication symbols. 
The $\mathrm{AP}_{\mathrm{tx}}$ and the CU coordinate to design the precoders $\mathbf{F}_m$, with the assumption that the CSI is perfectly estimated and locally known by each $\mathrm{AP}_{\mathrm{tx}}$.  
The transmitted data streams $s_k$ are independent and have unit power, i.e., \(\mathbb{E}[\mathbf{s}\mathbf{s}^H] = \mathbf{I}_{K}\). Additionally, the power budget for the beamforming vector is constrained by \(P_m\), formally
\begin{equation} \label{powerBudCF}
	||\mathbf{F}_m||_F^2 \leq P_m.
\end{equation}

\subsection{Communication Model}

As previously mentioned, the transmission of communication data follows a CF approach. Consequently, the signal received by the $k$th UE is modeled as the sum of the signals transmitted by the $M$ $\mathrm{AP}_{\mathrm{tx}}$, as described in \eqref{II_eq01_tx_signal}. The received signal \(y_k\) is given by
\begin{equation}
	\begin{aligned} \label{II_eq3_rx_UE_signal}
		y_{k} 	& = \sum_{m=1}^M\mathbf{h}_{m,k}^H\mathbf{x}_m + n_{k} \\
			    & = \underbrace{\sum_{m=1}^M \mathbf{h}_{m,k}^H \mathbf{f}_{m,k} s_k}_{\text{Desired Signal (DS)}} + \underbrace{\sum_{\substack{i=1 \\ i \neq k}}^K \sum_{m=1}^M  \mathbf{h}_{m,k}^H\mathbf{f}_{m,i} s_i}_{\text{Multi-user Interference (MUI)}} + \underbrace{n_k}_{\text{Noise}}
	\end{aligned}
\end{equation}
where \( n_k \sim \mathcal{CN}(0, \sigma_k^2) \) represents complex additive white Gaussian noise (AWGN) with zero mean and variance \( \sigma_k^2 \). The term \(\mathbf{h}_{m,k} \in \mathbb{C}^{N_{\mathrm{tx}}}\) denotes the communication channel between the \(m\)th $\mathrm{AP}_{\mathrm{tx}}$ and the \(k\)th UE.

The communication channel $\mathbf{h}_{m,k}$ is modeled as a narrowband block-fading propagation channel, as described in \cite{narrCite}. Formally, it is expressed as
\begin{equation} \label{channMo}
	\mathbf{h}_{m,k} = \frac{1}{\sqrt{L}}\sum_{l=1}^{L} \alpha_{m,k}^{(l)} \mathbf{a}_{N_{\mathrm{tx}}}(\psi_{m,k}^{(l)}), 
\end{equation}
where $L$ is the number of paths, $\alpha_{m,k}^{(l)}$ represents the channel gain, and $\mathbf{a}_{N_{\mathrm{tx}}}(\psi_{m,k}^{(l)}) \in \mathbb{C}^{N_{\mathrm{tx}}}$ represents the transmit array response vector of the $l$th path, with $\psi_{m,k}^{(l)}$ denoting the angle of departure (AoD). For an ULA with $N_{\mathrm{tx}}$ antenna elements, the array response vector in the direction of $\psi$ can be represented as
\begin{equation} 
	\mathbf{a}_{N_{\mathrm{tx}}}(\psi) = \left[ 1, e^{j\frac{2\pi}{\lambda} d \sin (\psi)}, \cdots, e^{j\frac{2\pi}{\lambda} d (N_{\mathrm{tx}} - 1) \sin (\psi)} \right]^T 
	\label{eqArray}
\end{equation}
where $\lambda$ and $d$ represent the signal wavelength and the antenna spacing, respectively.
 
Building on the established conditions of independence among data streams and between the data streams and receiver noise, we can now derive the total power of the received signal in \eqref{II_eq3_rx_UE_signal}. Under these assumptions,  the total power of the received signal in \eqref{II_eq3_rx_UE_signal} can be derived as
\begin{equation}
	\begin{aligned}
		\mathbb{E}[|y_k|^2] & =  \mathbb{E}[|\text{DS}|^2] +  \mathbb{E}[|\text{MUI}|^2] +  \mathbb{E}[|\text{Noise}|^2]\\
		& = P_{\text{DS},k} + P_{\text{MUI},k} + \sigma_k^2,
	\end{aligned}
\end{equation} 
where \( P_{\text{DS},k} \) and \( P_{\text{MUI},k} \) are defined as
\begin{equation} \label{desSig_CF}
	P_{\text{DS},k} = \left|\sum_{m=1}^M \mathbf{h}_{m,k}^H \mathbf{f}_{m,k}\right|^2,
\end{equation}
and
\begin{equation} \label{mui_CF}
	P_{\text{MUI},k} = \sum_{\substack{i=1 \\ i \neq k}}^K \left|\sum_{m=1}^M \mathbf{h}_{m,k}^H \mathbf{f}_{m,i}\right|^2.
\end{equation}
Using this decomposition, the SINR at the \( k \)th UE is given by
\begin{equation} 
	\label{II_eq06_SINR}
	\mathrm{SINR}_{k} = \frac{P_{\text{DS},k}}{P_{\text{MUI},k} + \sigma_k^2}.
\end{equation}
Here, the numerator corresponds to the power of the desired signal received by the \( k \)th UE, while the denominator accounts for the total interference from other users and the noise power at the receiver.

\subsection{Radio-Sensing Model}
This work considers a multi-static radar configuration, where the $\mathrm{AP}_{\mathrm{tx}}$s direct a radio-sensing probing beam in the direction of a point-like target. It is assumed that a LoS path exists both from the $\mathrm{AP}_{\mathrm{tx}}$ to the target and from the target to the $\mathrm{AP}_{\mathrm{rx}}$. Additionally, the channel between any $\mathrm{AP}_{\mathrm{tx}}$ and $\mathrm{AP}_{\mathrm{rx}}$ pair is restricted to this combination of LoS paths. Therefore, the channel between the $m$-th $\mathrm{AP}_{\mathrm{tx}}$ and the $n$-th $\mathrm{AP}_{\mathrm{rx}}$ consists of a single path \cite{Demirhan2023}, as shown in Fig. \ref{00_fig_1_multistatic_CF_isac_scenario}. The signal received at the $n$th $\mathrm{AP}_{\mathrm{rx}}$ is given by
\begin{equation}
	\mathbf{y}_n = \sum_m \hat\alpha_{m,n} \mathbf{a}_{N_{\textit{rx}}}(\phi_n) \mathbf{a}_{N_{\textit{tx}}}^H(\theta_m)\mathbf{x}_m + \mathbf{n}_n
\end{equation}
where $\hat\alpha_{m,n}$ is the path gain, which includes the effect of the path-loss, and radar cross section (RCS) of the target \cite{Demirhan2023}. The variables $\theta_m$ and $\phi_n$ represent the angle of departure (AoD) from the $m$th $\mathrm{AP}_{\mathrm{tx}}$ and the angle of arrival (AoA) at the $n$th $\mathrm{AP}_{\mathrm{rx}}$, respectively.
The signals received through the \(N_{\textrm{rx}}\) antennas of the \(n\)th \(\mathrm{AP}_{\mathrm{rx}}\) are locally combined as follows
\begin{equation} \label{recCfsig}
	\begin{aligned}
		y_n & = \mathbf{g}_n^H\mathbf{y}_n\\ 
		& = \gamma_n\sum_m \hat\alpha_{m,n}\mathbf{a}_{N_{\textit{tx}}}^H(\theta_m)\mathbf{F}_m\mathbf{s} + \mathbf{g}_n^H\mathbf{n}_n,
	\end{aligned}
\end{equation}
where $\mathbf{g}_n \in \mathbb{C}^{N_{\text{rx}}}$ represents the $n$th combiner, and $\gamma_n = \mathbf{g}_n^H\mathbf{a}(\phi_n)$. The signals in \eqref{recCfsig}, from the \(N\) $\mathrm{AP}_{\mathrm{rx}}$, are forwarded to the central unit (CU) and coherently added. Thus, the signal at the CU is given by
\begin{equation}
	\label{II_eq9_rx_signal_CU}
		y  =  \sum_n \sum_m  \gamma_n  \hat\alpha_{m,n} \mathbf{a}_{N_{\textit{tx}}}^H(\theta_m)\mathbf{F}_m\mathbf{s} + \sum_n\mathbf{g}_n^H\mathbf{n}_n.
\end{equation}

We consider that the noise at different $\mathrm{AP}_{\mathrm{rx}}$ are independent and identically distributed (\textit{i.i.d.}) random variables.  In addition, we employ the Swerling I model, which is useful for modeling targets with multiple scattering elements, where the scatterings between different transmitter/receiver pairs are independent \cite{modernRadar}. That means that there is no correlation between $\hat\alpha_{m,n}$ for the different $\mathrm{AP}_{\mathrm{tx},m}$/$\mathrm{AP}_{\mathrm{tx},n}$ pairs. Consequently, the expected power of \eqref{II_eq9_rx_signal_CU} can be obtained as
\begin{equation}
	\mathbb{E}[|y|^2] = \sum_m \sum_n \sigma_{m,n}^2 \left\|\gamma_n \mathbf{a}_{N_{\textit{tx}}}^H(\theta_m) \mathbf{F}_m \right\|_F^2 + \sum_n \left\|\mathbf{g}_n \right\|^2 \sigma_n^2.
\end{equation}
From the above, it follows that the sensing SNR (sSNR) is given by
\begin{equation} \label{sSNR_CF}
	\mathrm{sSNR} = \frac{\sum_m \sum_n \sigma_{m,n}^2 \left\|\gamma_n \mathbf{a}_{N_{\textit{tx}}}^H(\theta_m) \mathbf{F}_m \right\|_F^2}{\sum_n \left\|\mathbf{g}_n \right\|^2 \sigma_n^2},
\end{equation}
where $\sigma_{m,n}^2$ denotes the variance of the sensing channel.

\section{Distributed Beamforming Framework} \label{distBemFra}
 
In prior work on beamforming design for CF ISAC \cite{Demirhan2023, Mao2023GlobeCom, Mao2023, Cao2023CF, Liu2024}, the computation of the precoders is entirely performed at the CU, requiring significant data exchange between the $\mathrm{AP}_{\mathrm{tx}}$s and the CU. In this setup, the $\mathrm{AP}_{\mathrm{tx}}$s function as distributed antennas, leading to two main drawbacks. First, the computational complexity at the CU is increased due to the high dimensionality of the optimization variables. Second, fronthaul capacity limitations arise, as the CSI must be transmitted from the $\mathrm{AP}_{\mathrm{tx}}$s to the CU for beamforming design, followed by the transmission of the precoding matrices or signals from the CU back to the $\mathrm{AP}_{\mathrm{tx}}$s.
 
To mitigate these drawbacks, this section introduces a distributed beamforming framework in which the precoding vectors $\mathbf{f}_{m,k} \: \forall m,k$ are divided into two parts. Specifically, we define these precoding vectors as follows
\begin{equation} \label{distAPPCF}
	\mathbf{f}_{m,k} = \delta_{m,k} \mathbf{w}_{m,k},
\end{equation}
where $\mathbf{w}_{m,k} \in \mathbb{C}^{N_{\text{tx}}}$ represents the local precoding vector, computed at each $\mathrm{AP}_{\mathrm{tx}}$, and $\delta_{m,k}$ is a complex weight computed at the CU. Fig. \ref{01_fig_2_signal_formation} illustrates how the transmitted signal is formed. In the following, any reference to the precoder design refers to the design of the distributed beamforming framework. The formulation up to this point can be straightforwardly rewritten by simply substituting \eqref{distAPPCF}.

\begin{figure}[!htb]
	\centering
	\includegraphics[width=0.47\textwidth]{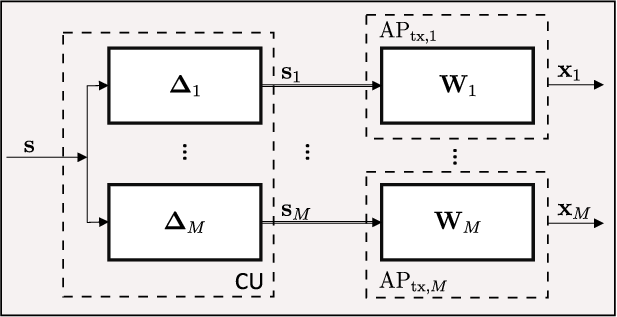}
	\caption{Diagram of the generation of the transmitted signal in a distributed paradigm within CF ISAC.}
	\label{01_fig_2_signal_formation}
\end{figure} 
 
As depicted in Fig. \ref{01_fig_2_signal_formation}, the data stream vector $\mathbf{s}$ is precoded at the CU using $M$ diagonal precoding matrices. The resulting signals are then transmitted through the fronthaul to each of the $M$ $\mathrm{AP}_{\mathrm{tx}}$. Formally, the signal sent from the CU to the $m$th $\mathrm{AP}_{\mathrm{tx}}$ is given by
\begin{equation} \label{modSignal}
	\mathbf{s}_m = \mathbf{\Delta}_m \mathbf{s},
\end{equation}
where \(\mathbf{\Delta}_m = \mathrm{diag}\{\delta_{m,1}, \ldots, \delta_{m,K}\} \in \mathbb{C}^{K \times K}\) represents the central precoding matrix of the $m$th $\mathrm{AP}_{\mathrm{tx}}$ processed by the CU. Then, \eqref{modSignal} is precoded using a local precoding matrix by each $\mathrm{AP}_{\mathrm{tx}}$. Therefore, the transmitted signal \eqref{II_eq01_tx_signal} from the $m$th $\mathrm{AP}_{\mathrm{tx}}$ can be reformulated as
\begin{equation} \label{modSignallocal} 
	\mathbf{x}_m = \mathbf{W}_m \mathbf{s}_m,
\end{equation}
where $\mathbf{W}_m \in \mathbb{C}^{N_{\mathrm{tx}} \times K}$ is the precoder corresponding to the $m$th $\mathrm{AP}_{\mathrm{tx}}$. From \eqref{modSignal} and \eqref{modSignallocal}, it follows that $\mathbf{F}_m = \mathbf{W}_m \mathbf{\Delta}_m$.

\section{Two-stage Distributed Beamforming Design Algorithm} \label{sectIII_CF}

This section introduces a two-stage distributed algorithm for designing the central and local precoders, \(\mathbf{\Delta}_m\) and \(\mathbf{W}_m\), respectively.
The problem is formulated as the maximization of the sum of the SINR \eqref{II_eq06_SINR} for the UEs, subject to the power budget for each \(\mathrm{AP}_{\mathrm{tx}}\) and the minimum sSNR requirements for radio-sensing purposes.
However, the proposed optimization problem is a challenging and computationally expensive non-convex problem. To address this, we introduce a solution based on canceling the inter-user interference.
The interference-free precoders are designed in a two-stage distributed manner between the \(\mathrm{AP}_{\mathrm{tx}}\)s and the CU. 
However, both problems to be solved at the \(\mathrm{AP}_{\mathrm{tx}}\)s and the CU remain non-convex. Therefore, we leverage a majorization-minimization (MM) framework to obtain a simpler lower-bound surrogate function, from which an iterative optimization algorithm can be derived to approximate a solution to the original problem.

\subsection{Problem Formulation}

The objective of this work is to design  \(\mathbf{\Delta}_m\), and \(\mathbf{W}_m\) that maximize the sum of the downlink SINR for the $K$ UEs in a CF ISAC system, considering a multi-user and single-target scenario. The solution must satisfy two constraints; the power budget per \(\mathrm{AP}_{\mathrm{tx}}\), and the sSNR in \eqref{sSNR_CF} must be greater than the minimum value. Consequently, the optimization problem can be formulated as follows
\begin{subequations} \label{III_P1}
	\begin{align}
		\mathop{\max }\limits_{\{\mathbf{W}_m, \mathbf{\Delta}_m\}_m} & \sum_{k} \text{SINR}_k \label{III_P1:a}\\  
		\mathrm{s.t.} & \quad \|\mathbf{W}_m \mathbf{\Delta}_m\|_F^2 \leq P_m, \quad \forall m \\
		& \quad \mathrm{sSNR} \geq \Delta, \label{III_P1:c}
	\end{align}
\end{subequations}
where $P_m$ denotes the available power budget at the $m$th \(\mathrm{AP}_{\mathrm{tx}}\), and \(\Delta\) represents the minimum threshold for the sSNR. This threshold is crucial for ensuring that the received signal power is sufficiently strong to enable accurate detection and estimation of the target parameters.

The MUI in \eqref{II_eq06_SINR} significantly increases the complexity of \eqref{III_P1}, making it computationally expensive to find a solution. This complexity arises from the need to jointly account for the interference among multiple users, which complicates the maximization of the SINR objective function in a CF scenario. To address this challenge, we propose a solution that simplifies the optimization problem by canceling the MUI. This approach effectively reduces the dimensionality of the optimization problem, allowing for a more tractable solution.

\subsection{Interference Cancellation via Null-space Projection}
One approach to eliminate the MUI \eqref{mui_CF} between the data streams involves projecting the precoder vectors, i.e., $\mathbf{w}_{m,k} \: \forall m,k$, into null-space (NS) of the interference channels. This projection is performed locally by each \(\mathrm{AP}_{\mathrm{tx}}\), without requiring any information exchange between the \(\mathrm{AP}_{\mathrm{tx}}\)s. As a result, the NS-based optimization problem can be reformulated as follows
\begin{subequations} \label{III_P2}
	\begin{align}
		\mathop{\max }\limits_{\{\mathbf{\hat W}_m, \mathbf{\Delta}_m\}_m} & \sum_{k} \hat{P}_{\text{DS},k} \label{III_P2:a}\\  
		\mathrm{s.t.} & \sum_k\left|\delta_{m,k}\right|^2\|\mathbf{\hat w}_{m,k}\|^2  \leq P_m, \forall m \\
		& \sum_m \sum_n \sum_k\sigma_{m,n}^2 \left| \gamma_n \delta_{m,k}\mathbf{a}_{m,k}^H\mathbf{\hat w}_{m,k}\right|^2  \geq \tilde{\Delta}, \label{III_P2:c}
	\end{align}
\end{subequations}
where 
\begin{equation} \label{desNSsig}
	\hat{P}_{\text{DS},k} = \left|\sum_m\delta_{m,k}\mathbf{\hat h}_{m,k}^H\mathbf{\hat{w}}_{m,k}\right|^2,
\end{equation}
$\mathbf{\hat h}_{m,k}$ and $\mathbf{a}_{m,k}$ are obtained by projecting $\mathbf{h}_{m,k}$ and $\mathbf{a}_{N_{\textit{tx}}}(\theta_m)$ into the null space of the interference channel $\mathbf{H}_{m,\overline{k}}$, defined as
\begin{equation} \label{equivChann_CF}
	\begin{aligned}
		\mathbf{\hat h}_{m,k}   & = \mathbf{P}_{m,k}^H\mathbf{h}_{m,k} \\
		\mathbf{a}_{m,k}  & = \mathbf{P}_{m,k}^H\mathbf{a}_{N_{\textit{tx}}}(\theta_m) 			
	\end{aligned}
\end{equation}
where $\mathbf{P}_{m,k} = \text{null}\left(\mathbf{H}_{m,\overline{k}}\right)$ represents an orthonormal basis spanning the null space of the interference channel
\begin{equation}
	\mathbf{H}_{m,\overline{k}} \triangleq [\mathbf{h}_{m,1}, \cdots ,\mathbf{h}_{m,k-1}, \mathbf{h}_{m,k+1}, \cdots ,\mathbf{h}_{m,K}]^H.
\end{equation}  
Moreover, $\mathbf{\hat w}_{m,k}$ denotes the projected precoder, from which $\mathbf{w}_{m,k} = \mathbf{P}_{m,k}\mathbf{\hat w}_{m,k}$ follows, and $\tilde{\Delta} = \sum_n \left\|\mathbf{g}_n \right\|^2 \sigma_n^2 \Delta$.

\subsection{Iterative Two-stage Distributed Beamforming Design Algorithm}

The design of $\mathbf{\hat{W}}_m$ and $\mathbf{\Delta}_m$ is realized in a distributed manner, where the $\mathbf{\hat{W}}_m$ are computed in parallel across the $M$ $\mathrm{AP}_{\mathrm{tx}}$s, and $\mathbf{\Delta}_m$ is determined at the CU. The data exchange flow between the $\mathrm{AP}_{\mathrm{tx}}$ and the CU is illustrated at a glance in Fig. \ref{01_fig_2_split_method}. Note that the presented schematic is conceptual, emphasizing the data generated by each entity rather than the physical architecture, which can be implemented using various topologies, such as star, mesh, or tree.  

\begin{figure}[!htb]
	\centering
	\includegraphics[width=0.45\textwidth]{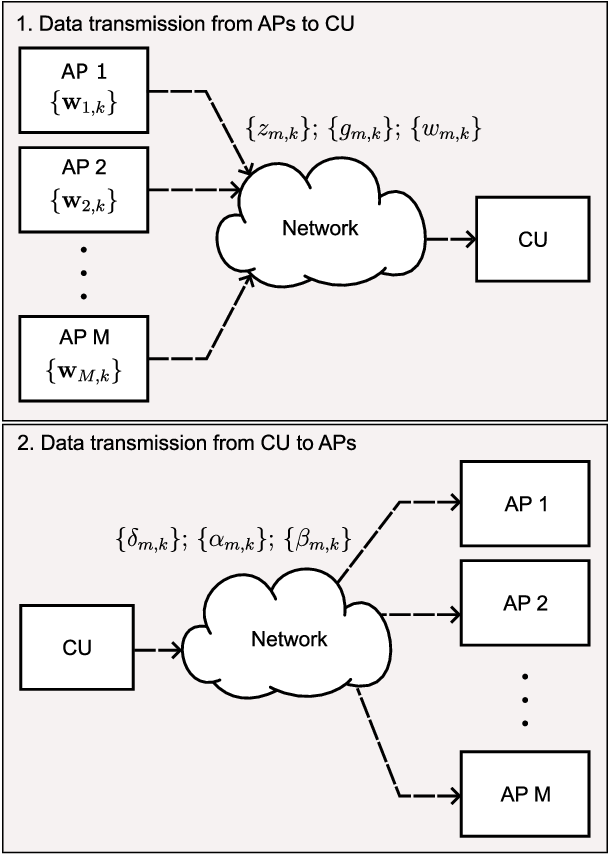}
	\caption{High-level schematic of the data exchange required for the two-stage distributed beamforming design algorithm.}
	\label{01_fig_2_split_method}
\end{figure} 

At the $m$th $\mathrm{AP}_{\mathrm{tx}}$, $\mathbf{\hat{W}}_m$ is computed using information sent by the CU, including the $\delta_{m,k}$ weights specific to that $\mathrm{AP}_{\mathrm{tx}}$. Additionally, data from the other $M-1$ $\mathrm{AP}_{\mathrm{tx}}$s is required, such as the $\alpha_{m,k}$ and $\beta_{m,k}$ values, which represent the equivalent communication and radio-sensing channels, which accounts local and central precoding, respectively. We want to highlight that the computed precoders remain local to each $\mathrm{AP}_{\mathrm{tx}}$, and only partial information is forwarded to the CU.
The CU computes the $\delta_{m,k}$ weights using partial information from the $\mathrm{AP}_{\mathrm{tx}}$, such as $z_{m,k}$ and $g_{m,k}$, which represent the equivalent local communication and radio-sensing channels, respectively. Additionally, $w_{m,k}$ denotes the power of $\mathbf{\hat{w}}_{m,k}$. 
The data exchange between $\mathrm{AP}_{\mathrm{tx}}$ and CU continues iteratively until the number of iterations $N_{\text{iter}}$ is reached.
The details about the required information by the $\mathrm{AP}_{\mathrm{tx}}$ and CU are provided in subsections \ref{sublocPrec} and \ref{subcentPrec}, respectively.

\subsection{Local Precoder Optimization at the $m$th $\mathrm{AP}_{\mathrm{tx}}$} \label{sublocPrec}

As previously mentioned, the optimization of $\mathbf{\hat{w}}_{m,k} \: \forall k$ is performed by the $m$th $\mathrm{AP}_{\mathrm{tx}}$, requiring partial information from the other $\mathrm{AP}_{\mathrm{tx}}$s and the CU. This problem is the same across all $\mathrm{AP}_{\mathrm{tx}}$s, allowing each to be resolved independently and in parallel. Therefore, the optimization problem at the $m$th $\mathrm{AP}_{\mathrm{tx}}$ is given by
\begin{subequations} \label{III_P4}
	\begin{align} 
		\mathop{\max }\limits_{\{\mathbf{\hat w}_{m,k}\}_k} & \sum_k\left|\delta_{m,k}\mathbf{\hat h}_{m,k}^H \mathbf{\hat w}_{m,k} + \sum_{i \neq m} \alpha_{i,k}\right|^2 \label{III_P4:a} \\  
		\mathrm{s.t.}  \quad &   \sum_k|\delta_{m,k}|^2\|\mathbf{\hat w}_{m,k}\|^2  \leq P_m\\
		\quad & \sum_k \sum_n \sigma_{m,n}^2 \left|\gamma_n \delta_{m,k}\mathbf{a}_{k}^H(\theta_m)\mathbf{\hat w}_{m,k}\right|^2 \nonumber \\
		& \quad + \sum_{i\neq m} \sum_k \sum_n \sigma_{i,n}^2 \left|\gamma_n \beta_{i,k} \right|^2 \geq \tilde{\Delta} \label{III_P4:c},
	\end{align}
\end{subequations}
where $\alpha_{i,k} = \delta_{i,k}z_{i,k}$ represents the equivalent communication channel, accounting for both central and local precoding. Here, $z_{i,k}$ denotes the equivalent local communication channel for the other $\mathrm{AP}_{\mathrm{tx}}$s, formally defined as ${z}_{i,k} = \mathbf{\hat h}_{i,k}^H \mathbf{\hat w}_{i,k}$. 
In the same way, $\beta_{i,k} = \delta_{i,k}g_{i,k}$ is the equivalent radio-sensing channel, accounting for both central and local precoding. The term $g_{i,k}$ refers to the local radio-sensing channel for the remaining $\mathrm{AP}_{\mathrm{tx}}$s, defined as ${g}_{i,k} = \mathbf{\hat a}_{i,k}^H \mathbf{\hat w}_{i,k}$. 

The optimization problem \eqref{III_P4} is non-convex due to the nature of the objective function and the radio-sensing related constraint \eqref{III_P4:c}. The maximization objective \eqref{III_P4:a} aims to maximize the sum of absolute values, which is a convex function, resulting in a non-convex optimization problem. On the other hand, the constraint in \eqref{III_P4:c} involves the sum of convex quadratic terms, but since the inequality is in the form of a lower bound, it represents the complement of a convex set. As a result, this constraint describes a non-convex region. 

An effective approach to addressing the previously non-convex problem is to leverage a minimization-majorization (MM) framework. This method allows us to generate a sequence of simpler, convex subproblems that are easier to solve, ultimately leading to an approximate solution for the original non-convex problem. The considered surrogate function is represented as the first-order Taylor polynomial of \eqref{III_P4:a}. The lower bound surrogate optimization problem can be recast as follows
\begin{subequations} \label{III_P5}
	\begin{align} 
		\{\mathbf{\hat w}_{m,k}^{(t)}\}_k = & \mathop{\arg\max }\limits_{\{\mathbf{\hat w}_{m,k}\}_k}  \:2 \Re \left(\sum_k \nabla_{\mathbf{\hat w}_{m,k}} \hat{P}_{\text{DS},k}^H\left(\mathbf{\hat w}_{m,k}^{(t-1)}\right) \mathbf{\hat w}_{m,k} \right)\\  
		\mathrm{s.t.}  \quad &  \sum_k|\delta_{m,k}|^2\|\mathbf{\hat w}_{m,k}\|^2  \leq P_m\\
		\quad & \frac{1}{\sqrt{MNK}} \Re \bigg(\sum_k \sum_n \sigma_{m,n} \gamma_n \delta_{m,k}\mathbf{a}_{k}^H(\theta_m)\mathbf{\hat w}_{m,k} \nonumber \\
		& \quad + \sum_{i\neq m} \sum_k \sum_n \sigma_{i,n} \gamma_n \beta_{i,k} \bigg)  \geq \tilde{\Delta}^{1/2} \label{III_P5:c}, 
	\end{align}
\end{subequations}
where the operator $\nabla_{\mathbf{\hat w}_{m,k}} \hat{P}_{\text{DS},k}$ denotes the gradient of \eqref{desNSsig} with respect to $\mathbf{\hat w}_{m,k}$, formally
\begin{equation}
	\nabla_{\mathbf{\hat w}_{m,k}} \hat{P}_{\text{DS},k} = |\delta_{m,k}|^2 \mathbf{\hat h}_{m,k} \mathbf{\hat h}_{m,k}^H \mathbf{\hat w}_{m,k} + \sum_{i \neq m} \alpha_{i,k} \delta_{m,k}^* \mathbf{\hat h}_{m,k},
\end{equation} 
\eqref{III_P5} serves as a simplified and convex approximation of \eqref{III_P4}. To handle the non-convexity of constraint \eqref{III_P4:c}, it was replaced by \eqref{III_P5:c}. This new constraint is convex but it is more restrictive since the
set defined by this constraint is a subset of \eqref{III_P4:c}, accordingly to proposition \ref{ch5__split_prep1} \cite{leyva2024hybrid}. 
\begin{prop} \label{ch5__split_prep1}
	Let us define $\mathcal{B} = \{\mathbf{\hat w}_{m,k} \forall k:    \eqref{III_P4:c}\}$ and $\mathcal{\hat{B}} = \{\mathbf{\hat w}_{m,k} \forall k:    \eqref{III_P5:c}\}$, then $\mathcal{\hat{B}} \subseteq \mathcal{B}$.
\end{prop}
\begin{proof}
	Sets $\mathcal{B}$ and $\mathcal{\hat{B}}$ define the points satisfying constraints \eqref{III_P4:c} and \eqref{III_P5:c}. As $\|\mathbf{x}\| \in \mathbb{C}^Q$ is a convex function, it can be lower-bounded by the linear function $\Re(\mathbf{x}_0^H\mathbf{x}/\|\mathbf{x}_0\|)$, which is valid for any $\mathbf{x}_0$. For $\mathbf{x}_0 = \mathbf{1}$ follows the inequality $\|\mathbf{x}\| \geq \Re(\mathbf{1}^T\mathbf{x}/\|\mathbf{1}\|) = 1/\sqrt{Q} \sum_q\Re(\mathbf{x}^H\mathbf{e}_q)$. If $\mathbf{x} = [x_1, \cdots, x_Q]$, where $x_q$ represents the $q$th term in  \eqref{III_P4:c}, then proposition \ref{ch5__split_prep1} follows.
\end{proof}
From the previous proposition, it follows that the optimal value of the original problem is less than or equal to the value of the new problem, while the solution satisfies the original constraints. The optimization problem \eqref{III_P5} is now convex and can be efficiently solved using convex optimization solver tools such as \texttt{CVX} \cite{cvx}.

\subsection{Precoder Optimization at the CU}
\label{subcentPrec}

The values of the local equivalent communication and radio-sensing channels, \( z_{m,k} \) and \( g_{m,k} \), along with the power of the precoder vectors defined as \( w_{m,k} = \|\mathbf{\hat w}_{m,k}\|^2 \), are transmitted to the CU by the \( M \) \( \mathrm{AP}_{\mathrm{tx}} \). With this information, the CU computes the weights \( \delta_{m,k} \) that solve \eqref{III_P2}. Hence, \eqref{III_P2} can be recast at the CU as
\begin{subequations} \label{III_P6}
	\begin{align} 
		\mathop{\max }\limits_{\{\delta_{m,k}\}_{m,k}} & \sum_k \left|\sum_m \delta_{m,k}z_{m,k}\right|^2   \\  
		\mathrm{s.t.}  &\quad   \sum_k w_{m,k} |\delta_{m,k}|^2  \leq P_m \quad \forall m  \\
		& \quad  \sum_m \sum_n \sum_k \sigma_{m,n}^2 \left|\gamma_n \delta_{m,k} {g}_{m,k}\right|^2 \geq \tilde \Delta.  
	\end{align}
\end{subequations}

Notice that the \eqref{III_P6} is similar to \eqref{III_P4}.  Hence, the lower-bound surrogate optimization problem can be reformulated as follows  
\begin{subequations} \label{III_P7}
	\begin{align} 
		 \{\delta_{m,k}^{(t)}\}_{m,k} 	& = \nonumber \\  
										& \hspace{-.3 cm } \mathop{\arg\max }\limits_{\{\delta_{m,k}\}_{m,k}} \:  2\Re\left(\sum_m \sum_k \nabla_{\delta_{m,k}} \hat{P}_{\text{DS},k}^H\left(\delta_{m,k}^{(t-1)}\right) \delta_{m,k} \right)  \label{III_P7:a} \\  
					 \mathrm{s.t.}  \:  & \sum_k w_{m,k} |\delta_{m,k}|^2  \leq P_m \quad \forall m  \\ 
										& \frac{1}{\sqrt{MNK}}\Re \left(\sum_m \sum_n \sum_k\sigma_{m,n} \gamma_n \delta_{m,k} {g}_{m,k}\right) \nonumber \\ 
										& \hspace{5cm}  \geq \tilde \Delta^{1/2}    
	\end{align}
\end{subequations}
where $\nabla_{\delta_{m,k}} \hat{P}_{\text{DS},k}$ represents the gradient of \eqref{desNSsig} with respect to $\delta_{m,k}$, formally
\begin{equation}
	\nabla_{\delta_{m,k}} \hat{P}_{\text{DS},k} = \delta_{m,k} |z_{m,k}|^2 + z_{m,k}^* \sum_{i \neq m} \delta_{i,k}z_{i,k},
\end{equation} 
The optimization problem \eqref{III_P7} is convex and can be efficiently solved using for example the convex optimization library \texttt{CVX} \cite{cvx}. The proposed two-stage distibuted beamforming design algorithm is outlined in Algorithm \ref{alg:Alg1_CF}.
\begin{algorithm}[]
		\caption{Two-stage Distributted Beamforming Design Algorithm}
		\label{alg:Alg1_CF}
		\SetAlgoLined
		\KwIn{$\theta_{m}, P_m \: \forall m$}
		\KwOut{$\mathbf{W}_m, \mathbf{\Delta}_m \: \forall m$}
		1. The APs estimate $\mathbf{h}_{m,k} \: \forall m,k$\;
		2. The APs compute the equivalent channels that eliminate the inter-user interference as in \eqref{equivChann_CF}\;
		3. The APs randomly generates $\mathbf{\hat w}_{m,k}^{(0)} \: \forall m,k$\; 
		4. Initialize $\delta_{m,k} = 1 \: \forall m,k$,  $z_{m,k} = 0 \: \forall m,k$, $g_{m,k} = \frac{1}{M} \sqrt{\frac{\mathbf{\tilde \Delta}}{K}} \: \forall m,k$\;
		\For{$j = 1:N_{\mathrm{iter}}$}{
			5. In parallel, the $M$ $\mathrm{AP}_{\mathrm{tx}}$ compute $\mathbf{\hat w}_{m,k}$ as in \eqref{III_P4}\;
			6. The $M$ $\mathrm{AP}_{\mathrm{tx}}$ forwards $w_{m,k}, z_{m,k}, g_{m,k}$ to the CU\;  
			7. The CU computes $\delta_{m,k}, \alpha_{m,k}, \beta_{m,k}$ as in \eqref{III_P6} and forwards them to the $M$ $\mathrm{AP}_{\mathrm{tx}}$;
		}
		\KwRet{$\mathbf{w}_{m,k}$, $\delta_{m,k}$}
\end{algorithm}

Algorithm \ref{alg:Alg1_CF} computes the precoders through an iterative cooperative approach between the $\mathrm{AP}_{\mathrm{tx}}$s and the CU in a CF ISAC system.
Initially, each $\mathrm{AP}_{\mathrm{tx}}$ estimates the CSI and computes the equivalent channels to eliminate inter-user interference. The APs initialize variables in step (4) and compute the precoders $\mathbf{\hat w}_{m,k}$. These precoders are then used to update parameters in step (6) and are forwarded to the CU.
The CU computes $\delta_{m,k}$, which are then used to update $\alpha_{m,k}$ and $\beta_{m,k}$ in step (7) before being sent back to the  $\mathrm{AP}_{\mathrm{tx}}$s.
The algorithm iteratively refines the central and local precoders until the maximum number of iterations, \( N_{\mathrm{iter}} \), between the $\mathrm{AP}_{\mathrm{tx}}$s and the CU is reached.

\subsection{Fronthaul Load and Computational Complexity Analysis}

This subsection analyzes the fronthaul load and computational complexity at the CU of our proposal. Additionally, we compare the fronthaul load with existing contributions in the literature, while the computational complexity is compared to the centralized beamforming solution of \eqref{III_P1}
\footnote{Notice that the optimization problem in \eqref{III_P1} is the same non-convex problem our proposal addresses. Similar to the two-stage distributed beamforming algorithm, the non-convexity of the centralized approach is tackled using the MM technique; however, it involves a higher dimensionality of the optimization variables. Moreover, since the main focus of this paper is the distributed signal processing approach, detailed insights into the centralized optimization problem's solution are not provided.}. 

As depicted in Fig. \ref{01_fig_2_split_method}, the proposed two-stage distributed beamforming design algorithm requires the exchange of $3MK$ complex scalars between the $\mathrm{AP}_{\mathrm{tx}}$s and the CU, and another $3MK$ complex scalars in the opposite direction, i.e., from the CU to the $\mathrm{AP}_{\mathrm{tx}}$s, during each iteration. In the $\mathrm{AP}_{\mathrm{tx}}$-to-CU direction, these scalars include $z_{m,k}$, $g_{m,k}$, and $w_{m,k}$, while in the CU-to-$\mathrm{AP}_{\mathrm{tx}}$ direction, they include $\delta_{m,k}$, $\alpha_{m,k}$, and $\beta_{m,k}$. For $N_{\mathrm{iter}}$ iterations, the total number of complex scalars exchanged is $6N_{\mathrm{iter}}MK$. It is important to note that interchanged data values, as previously defined, are independent of $N_{\mathrm{tx}}$.
In contrast, centralized approaches \cite{Demirhan2023, Mao2023GlobeCom, Mao2023, Cao2023CF, Liu2024} require the transmission of the full CSI between the $\mathrm{AP}_{\mathrm{tx}}$s and the CU, as well as the precoders in the opposite direction. This results in the transmission of $2N_{\text{tx}}MK$ complex scalars, without the need for multiple iterations.
The independence from $N_{\mathrm{tx}}$ in our proposal significantly reduces the fronthaul load, particularly in CF massive MIMO systems, where access points may be equipped with numerous antennas. Consequently, the two-stage distributed solution effectively reduces the data exchange between the APs and the CU, enhancing system efficiency while minimizing fronthaul requirements.

The complexity of the optimization problems for the two-stage distributed beamforming design algorithm and the centralized beamforming algorithm, both solved using \texttt{CVX}, differ significantly due to the number of optimization variables. In the case of the two-stage distributed beamforming design, the optimization problem at the CU involves \(MK\) complex variables. The complexity per iteration of the interior-point method for this problem scales as \(O(M^3K^3)\), and with a total number of iterations scaling with \(O(\sqrt{MK})\), the overall complexity is approximately \(O(M^3K^3\sqrt{MK})\). In contrast, the centralized beamforming algorithm's optimization problem involves a much larger number of variables, specifically \(N_{\mathrm{tx}}MK\) complex variables, where \(N_{\mathrm{tx}}\) is the number of antennas at each $\mathrm{AP}_{\mathrm{tx}}$. The per iteration complexity in this case scales as \(O(N_{\mathrm{tx}}^3 M^3 K^3)\), and the total number of iterations is proportional to \(O(\sqrt{N_{\mathrm{tx}}MK})\), resulting in an overall complexity of \(O(N_{\mathrm{tx}}^3 M^3 K^3 \sqrt{N_{\mathrm{tx}}MK})\). The dependence on \(N_{\mathrm{tx}}\) leads to a significant increase in complexity, particularly for systems with large antenna arrays.
Thus, the two-stage distributed beamforming algorithm is far more computationally efficient than the centralized approach, especially in CF massive MIMO systems with many antennas per AP. The distributed method's independence from \(N_{\mathrm{tx}}\) makes it a more scalable and practical solution when minimizing fronthaul load and computational demands at the CU.

\section{Numerical Results} \label{sectIV_CF}

This section demonstrates the effectiveness of the proposed two-stage distributed beamforming algorithm (TsDBA). Additionally, the TsDBA is compared to the centralized beamforming solution of \eqref{III_P1}. As stated in footnote 1, the non-convexity is addressed by relying on the MM technique.
In the centralized solution, the beamforming design is entirely performed at the CU.
It is important to note that, to the best of the authors knowledge, there are no existing centralized contributions comparable to the TsDBA. The centralized methods proposed in \cite{Demirhan2023, Mao2023GlobeCom, Mao2023, Liu2024} are radar-centric, meaning they primarily focus on optimizing radar performance. In contrast, our method is communication-centric, prioritizing communication performance while maintaining sensing capabilities. 
Additionally, the communication-centric approach in \cite{Cao2023CF} uses the radar estimation rate as a metric for constraints, which encompasses latency and Doppler estimation rates in its computation. Conversely, the radar-sensing metric considered as a constraint in this work is the sSNR, which is independent of the latency and Doppler estimation rates. 
In summary, the differences with the aforementioned contributions render these methods incomparable to the TsDBA.

The following results were obtained for a CF ISAC scenario with \( M = 4 \) \(\mathrm{AP}_{\mathrm{tx}}\), \( N = 2 \) \(\mathrm{AP}_{\mathrm{rx}}\), and \( K = 2 \) UEs. The angles of the target relative to the \(M\) \(\mathrm{AP}_{\mathrm{tx}}\) are \( \theta_1 = -15^{\circ} \), \( \theta_2 = 35^{\circ} \), \( \theta_3 = 5^{\circ} \), and \( \theta_4 = 40^{\circ} \). Similarly, the angles from \(\mathrm{AP}_{\mathrm{rx}}\) are \( \phi_1 = 10^{\circ} \) and \( \phi_2 = -20^{\circ} \).
The $\mathrm{AP}_{\mathrm{tx}}$ and $\mathrm{AP}_{\mathrm{rx}}$ are equipped with ULAs, holding $N_{\mathrm{tx}} = 32$ and $N_{\mathrm{rx}} = 32$ antenna elements, respectively, spaced by $\lambda/2$. 
The communication channel is modeled according to \eqref{channMo}, with the number of paths set to $10$. In addition, $\psi_{m,k}^{(l)}$ is assumed to follow a uniform distribution over the range $[-\pi/2, \pi/2]$, while $\alpha_{m,k}^{(l)}$ follows a complex Gaussian distribution with a mean of zero and unit variance.
Regarding the sSNR, the noise variance at the $\mathrm{AP}_{\mathrm{rx}}$s is set to $\sigma_{n}^2 = -20$ dB. 
In addition, the $\mathrm{AP}_{\mathrm{rx}}$s are assumed to know the AoA (i.e., $\phi_n : \forall n$). A Maximum Ratio Combining (MRC) equalizer is employed at the $\mathrm{AP}_{\mathrm{rx}}$s primarily for its simplicity; while other selections could be made, the MRC is chosen for its straightforward implementation. The MRC is defined as 
\begin{equation} \mathbf{g}n = 				\frac{1}{\sqrt{N{\text{rx}}}}\mathbf{a}(\phi_n). 
\end{equation}
For simplicity, we assume that the variance of $\hat{\alpha}_{m,n} : \forall m,n$ is $\sigma_{m,n}^2 = -10$ dB.
Finally, the power budget of the distributed beamforming matrices is set to $P_m = 1 \: \forall m$.  

The remainder of this section analyzes the convergence of the proposed TsDBA. We then examine the impact of the sSNR on the transmit beampattern, followed by a study of the relationship between the sSNR constraint and the achievable sum SNR. Finally, we analyze the fronthaul load.

Fig. \ref{res01_CF} illustrates the convergence of the TsDBA by plotting the mean sum SINR as a function of the number of iterations, with \(N_{\text{iter}} = 10\) and \(\Delta \in \{30, 40\}\) dB. The results represent the mean over \(100\) random realizations of the communication channels \(\mathbf{w}_{m,k} \: \forall m,k\).

\begin{figure}[!htb]
	\centering
	\includegraphics[width=0.49\textwidth]{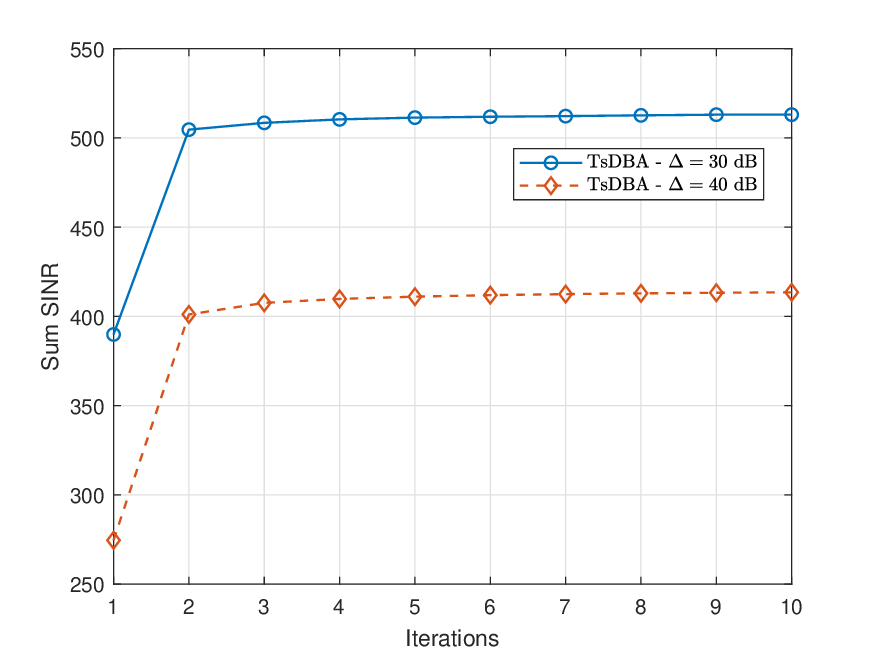}
	\caption{Convergence of TsDBA under radio-sensing constraints for $\Delta \in \{30, 40\}$ dB.}
	\label{res01_CF}
\end{figure} 
 
Fig. \ref{res01_CF} shows that the algorithm converges rapidly within the first few iterations, reporting higher sum SINR for the less stringent radio-sensing constraint \({\Delta} = 30 \, \text{dB}\) compared to ${\Delta} = 40$ dB, which reflects the trade-off between radio-sensing and communication SINR performance. Notably, after just $3$ iterations, the algorithm achieves approximately $99$\% of the maximum value for ${\Delta} = 30$ dB and around $98$\% for ${\Delta} = 40$ dB, indicating that a small number of iterations is sufficient to reach the maximum performance. Thus, $3$ iterations is an appropriate choice to balance computational complexity and convergence accuracy.

Fig. \ref{res02_CF} displays the transmitted beampattern obtained for $\mathrm{AP}_{\mathrm{tx},1}$ and $\mathrm{AP}_{\mathrm{tx},2}$ under two sSNR constraints, i.e., \(\Delta \in \{40, 46\}\) dB. Here, \(46\) dB is the maximum value of \(\Delta\) that renders the problem feasible under the specified parameters (i.e., \(N_{\text{tx}} = 32\), \(P_m = 1\)). The scenario involves $K = 2$ UEs, with the channel randomly generated using \eqref{channMo}. For both constraints, the same communication channel was considered.
\vspace{-.3cm}
\begin{figure}[!htb]
	\centering
	\begin{subfigure}[b]{0.49\textwidth} 
		\centering
		\includegraphics[width=\textwidth]{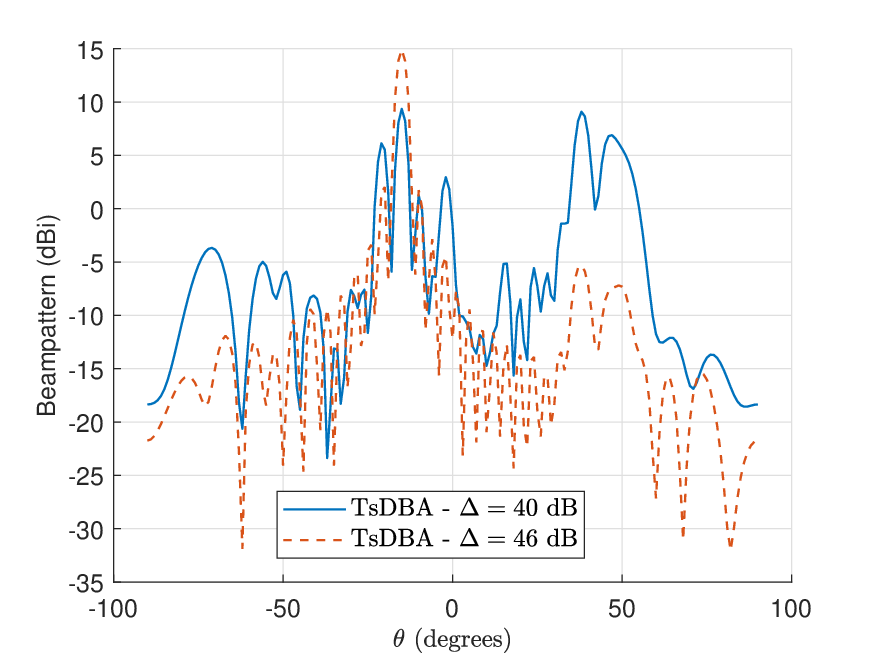}
		\caption{}
		\label{fig:ap1}
	\end{subfigure}
	\begin{subfigure}[b]{0.49\textwidth}
		\centering
		\includegraphics[width=\textwidth]{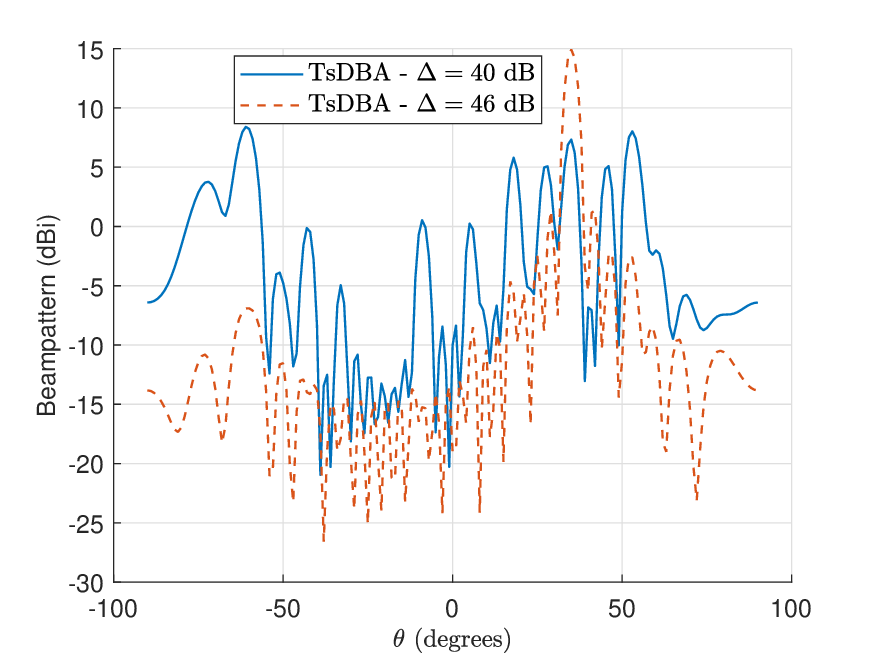}
		\caption{}
		\label{fig:ap2}
	\end{subfigure}
	\caption{Transmit beampatterns obtained by the TsDBA. (a) transmit beampattern of \(\mathrm{AP}_{\mathrm{tx},1}\), with $\theta_1 = -15^{\circ}$; (b) transmit beampattern of \(\mathrm{AP}_{\mathrm{tx},2}\), with $\theta_2 = 35^{\circ}$.}
	\label{res02_CF}
\end{figure}

The results illustrate the relationship between \(\Delta\) and the transmitted beampattern for \(\mathrm{AP}_{\mathrm{tx},1}\) and \(\mathrm{AP}_{\mathrm{tx},2}\). The beampatterns for \(\mathrm{AP}_{\mathrm{tx},3}\) and \(\mathrm{AP}_{\mathrm{tx},4}\) are not shown, but their results are similar to those of \(\mathrm{AP}_{\mathrm{tx},1}\) and \(\mathrm{AP}_{\mathrm{tx},2}\). From Fig. \ref{res02_CF}, it is evident that as the value of \(\Delta\) increases, the gain of the transmitted beampattern toward the target also increases. Furthermore, Fig. \ref{res02_CF} demonstrates that for \(\Delta = 46\) dB, the transmitted beampattern shows a significant reduction in the power of the lateral lobes, which correspond to directions away from the target. This suggests that less signal power is being directed toward the UEs, resulting in a decrease in the signal power received by them. The following result examines this trade-off and compares it with the centralized solution.

Fig. \ref{res03_CF} shows the trade-off between the mean of the sum of the SINR and the sSNR constraint $\Delta$. The sum of the SINR is calculated for $\Delta$ values ranging from 30 to 44 dB. For the TsDBA, we have set $N_{\text{iter}} = 3$, as this was shown to be an appropriate choice.
\begin{figure}[!htb]
	\centering
	\includegraphics[width=0.49\textwidth]{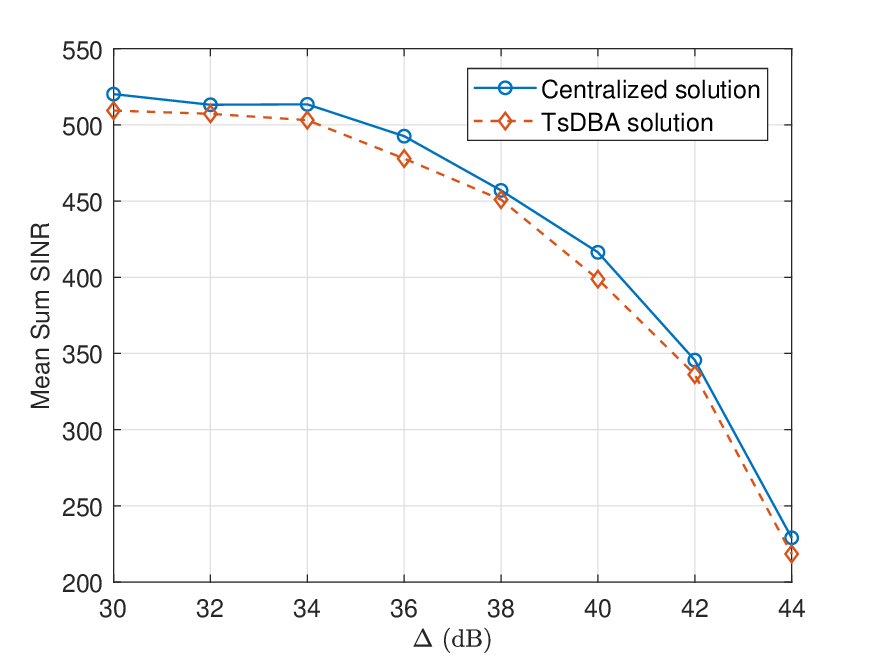}
	\caption{The trade-off between the sSNR constraint value $\Delta$ and the sum of the SINR, for CF-ISAC scenario with $M = 4$, $N = 2$, and $K = 2$.}
	\label{res03_CF}
\end{figure} 

The results presented in Fig. \ref{res03_CF} indicate that the proposed TsDBA solution demonstrates a performance close to the centralized solution. Also, both solutions exhibit a decreasing trend in mean sum SINR as $\Delta$ increases. This trend highlights a trade-off between the radio-sensing constraints and the achievable sum SINR. Although the centralized solution slightly outperforms the TsDBA across $\Delta$ values, the TsDBA significantly reduces fronthaul load requirements.

Fig. \ref{res04_CF} compares the fronthaul load required for beamforming design between the centralized and TsDBA solutions as a function of $N_{\text{tx}}$. The results were obtained using $N_{\text{iter}} = 3$, with $M \in \{4,16\}$ $\mathrm{AP}_{\mathrm{tx}}$ and $K \in \{2, 8\}$ UEs. 
\begin{figure}[!htb]
	\centering
	\includegraphics[width=0.49\textwidth]{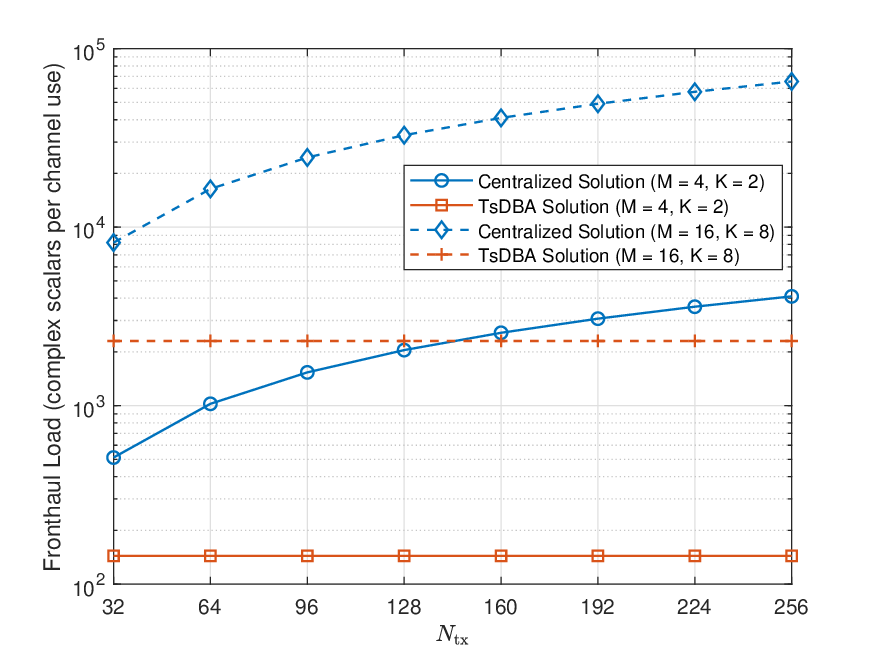}
	\caption{Fronthaul load comparison for beamforming design between the TsDBA and centralized solutions.}
	\label{res04_CF}
\end{figure} 
The fronthaul load of the centralized solution increases with \(N_{\text{tx}}\), leading to significant data transmission demands between the \(\mathrm{AP}_{\mathrm{tx}}\)s and the CU. In contrast, the TsDBA demonstrates a flat curve, indicating that the fronthaul load remains constant, regardless of \(N_{\text{tx}}\).
This difference arises because, in the TsDBA, the \(\mathrm{AP}_{\mathrm{tx}}\)s transmit equivalent communication and radio-sensing channels to the CU, rather than the full CSI required by the centralized solution.

\section{Conclusions} \label{sectV_CF}
This paper introduced a novel two-stage distributed beamforming design algorithm for the CF ISAC paradigm. By distributing the signal processing tasks between the \(\mathrm{AP}_{\mathrm{tx}}\)s and the CU, our approach effectively reduces the fronthaul load, addressing one of the key challenges in scaling CF-ISAC systems. The proposed method optimizes the sum of SINR for communication users while satisfying per-AP power constraints and sSNR requirements for sensing tasks. We formulated the resulting non-convex optimization problems and solved them using an iterative MM algorithm, which decomposes the problem into simpler, convex subproblems. Our results demonstrate that the two-stage distributed beamforming design algorithm achieves comparable performance to centralized solution, with the added benefits of reduced computational complexity and significantly reduced fronthaul load. The proposed two-stage distributed beamforming design algorithm provides a promising solution for enhancing the efficiency and scalability of CF-ISAC systems.

\section*{Acknowledgments}
This work has received funding from the FCT - Fundação para a Ciência e a Tecnologia under the PhD Research Studentships 2022.12379.BD and REVOLUTION project 2022.08005.PTDC, and from FCT/MCTES through national funds and when applicable co-funded EU funds under the project UIDB/50008/2020-UIDP/50008/2020. Additionally, this work has received funding from the European Union’s Horizon Europe research under the Smart Networks and Services Joint Undertaking (SNS JU) project 6GMUSICAL, grant agreement No. 101139176.

%
%
%
%
%
%
%
%
%


\begin{thebibliography}{1}
\bibliographystyle{IEEEtran}

\bibitem{paul2017}
B. Paul, A. R. Chiriyath and D. W. Bliss, "Survey of RF communications and sensing convergence research," \textit{IEEE Access}, vol. 5, pp. 252-270, Dec. 2017.

\bibitem{liu2020} 
F. Liu, C. Masouros, A. P. Petropulu, H. Griffiths and L. Hanzo, "Joint radar and communication design: applications, state-of-the-art, and the road ahead," \textit{IEEE Trans. Commun.}, vol. 68, no. 6, pp. 3834-3862, Jun. 2020.

\bibitem{leyva2021} 
L. Leyva, D. Castanheira, A. Silva, A. Gameiro and L. Hanzo, "Cooperative Multiterminal Radar and Communication: A New Paradigm for 6G Mobile Networks," \textit{IEEE Veh. Technol. Mag.} vol. 16, no. 4, pp. 38-47, Dec. 2021.

\bibitem{cui2021} 
Y. Cui, F. Liu, X. Jing, and J. Mu, "Integrating Sensing and Communications for Ubiquitous IoT: Applications, Trends, and Challenges," \textit{IEEE Netw.}, vol. 35, no. 5, pp. 158-167, Oct. 2021.

\bibitem{zhang2022} 
J. A. Zhang et al., "Enabling Joint Communication and Radar Sensing in Mobile Networks—A Survey," \textit{IEEE Commun. Surv. Tutor.}, vol. 24, no. 1, pp. 306-345, Oct. 2021.

\bibitem{liu2023} 
F. Liu et al., "Seventy Years of Radar and Communications: The road from separation to integration," \textit{IEEE Signal Process. Mag.}, vol. 40, no. 5, pp. 106-121, Jul. 2023.

\bibitem{3gpp2023tr}
3GPP, "Feasibility Study on Integrated Sensing and Communication, TR 22.837 (Release 19)," 2023. [Online]. Available: \url{https://portal.3gpp.org/desktopmodules/Specifications/SpecificationDetails.aspx?specificationId=4044}.

\bibitem{3gpp2023ts}
3GPP, "Service requirements for Integrated Sensing and Communication, TS 22.137 (Release 19)," 2023. [Online]. Available: \url{https://portal.3gpp.org/desktopmodules/Specifications/SpecificationDetails.aspx?specificationId=4198}.

\bibitem{giordani2020toward} 
M. Giordani, M. Polese, M. Mezzavilla, S. Rangan, and M. Zorzi, "Toward 6G Networks: Use Cases and Technologies," \emph{IEEE Commun. Mag.}, vol. 58, no. 3, pp. 55-61, Mar. 2020.

\bibitem{liu2022isc} 
F. Liu \emph{et al.}, "Integrated Sensing and Communications: Toward Dual-Functional Wireless Networks for 6G and Beyond," \emph{IEEE J. Sel. Areas Commun.}, vol. 40, no.6, pp.1728-1767, Jun. 2022.


\bibitem{liu18}
F. Liu, C. Masouros, A. Li, H. Sun and L. Hanzo, "MU-MIMO Communications With MIMO Radar: From Co-Existence to Joint Transmission," \textit{IEEE Trans. Wirel. Commun.}, vol. 17, no. 4, pp. 2755-2770, Apr. 2018.

\bibitem{eldar20}
X. Liu, T. Huang, N. Shlezinger, Y. Liu, J. Zhou and Y. C. Eldar, "Joint Transmit Beamforming for Multiuser MIMO Communications and MIMO Radar," \textit{IEEE Trans. Signal Process.}, vol. 68, pp. 3929-3944, Jun. 2020.

\bibitem{leyva2024hybrid}
L. Leyva, D. Castanheira, A. Silva, and A. Gameiro, "Hybrid Beamforming Design for Communication-Centric ISAC," \emph{IEEE Sens. J.}, vol. 24, no. 13, pp. 21179-21190, May. 2024.

\bibitem{Li2022} 
Y. Li and M. Jiang, "Joint Transmit Beamforming and Receive Filters Design for Coordinated Two-Cell Interfering Dual-Functional Radar-Communication Networks," \textit{IEEE Trans. Veh. Technol.}, vol. 71, no. 11, pp. 12362-12367, Nov. 2022. 

\bibitem{Chen2023} 
L. Chen, X. Qin, Y. Chen and N. Zhao, "Joint Waveform and Clustering Design for Coordinated Multi-Point DFRC Systems," \textit{IEEE Trans. Commun.}, vol. 71, no. 3, pp. 1323-1335, Mar. 2023. 

\bibitem{Cheng2024} 
G. Cheng, Y. Fang, J. Xu and D. W. K. Ng, "Optimal Coordinated Transmit Beamforming for Networked Integrated Sensing and Communications," \textit{IEEE Trans. Wirel. Commun.}, vol. 23, no. 8, pp. 8200-8214, Aug. 2024. 

 \bibitem{ngo2017cellfree} 
H. Q. Ngo, A. Ashikhmin, H. Yang, E. G. Larsson, and T. L. Marzetta, "Cell-Free Massive MIMO Versus Small Cells," \textit{IEEE Trans. Wirel. Commun.}, vol. 16, no. 3, pp. 1834-1850, Mar. 2017. 


\bibitem{Behdad2022} 
Z. Behdad, Ö. T. Demir, K. W. Sung, E. Björnson and C. Cavdar, "Power Allocation for Joint Communication and Sensing in Cell-Free Massive MIMO," \textit{GLOBECOM 2022 - 2022 IEEE Global Communications Conference}, Rio de Janeiro, Brazil, Jan. 2023. 

\bibitem{Behdad2023} 
Z. Behdad, Ö. T. Demir, K. W. Sung, E. Björnson and C. Cavdar, "Multi-Static Target Detection and Power Allocation for Integrated Sensing and Communication in Cell-Free Massive MIMO," \textit{IEEE Trans. Wirel. Commun.}, vol. 23, no. 9, pp. 11580-11596, Sept. 2024. 

\bibitem{Cao2023} 
Y. Cao and Q. -Y. Yu, "Joint Resource Allocation for User-Centric Cell-Free Integrated Sensing and Communication Systems," \textit{IEEE Commun. Lett.}, vol. 27, no. 9, pp. 2338-2342, Sept. 2023. 

\bibitem{Demirhan2023} U. Demirhan and A. Alkhateeb, "Cell-Free ISAC MIMO Systems: Joint Sensing and Communication Beamforming," arXiv:2301.11328, Feb. 2024. [Online]. Available: https://arxiv.org/abs/2301.11328

\bibitem{Mao2023GlobeCom}
W. Mao, Y. Lu, J. Liu, B. Ai, Z. Zhong and Z. Ding, "Beamforming Design in Cell-Free Massive MIMO Integrated Sensing and Communication Systems," \textit{GLOBECOM 2023 - 2023 IEEE Global Communications Conference}, Kuala Lumpur, Malaysia, 2023.

\bibitem{Mao2023} W. Mao, Y. Lu, C. -Y. Chi, B. Ai, Z. Zhong and Z. Ding, "Communication-Sensing Region for Cell-Free Massive MIMO ISAC Systems," \textit{IEEE Trans. Wirel. Commun.}, early access, Apr. 2024. 

\bibitem{Liu2024}
S. Liu, R. Liu, M. Li, and Q. Liu, "Cooperative Cell-Free ISAC Networks: Joint BS Mode Selection and Beamforming Design," arXiv preprint: 2305.10800, 2024. [Online]. Available: https://arxiv.org/abs/2305.10800

\bibitem{Cao2023CF} Y. Cao and Q. -Y. Yu, "Design and Performance Analyses of V-OFDM Integrated Signal for Cell-Free Massive MIMO Joint Communication and Radar System," \textit{IEEE Systems Journal}, vol. 17, no. 4, pp. 5943-5954, Dec. 2023

\bibitem{narrCite}
O. Ayach, S. Rajagopal, S. Surra, Z. Piand and R. Heath, "Spatially Sparse Precoding in millimeter wave MIMO systems," \textit{IEEE Trans. Wireless Commun.}, Vol. 13, no. 3, p. 1499–1513, Mar. 2014.

\bibitem{modernRadar}
M. A. Richards, J. Scheer, W. A. Holm, and W. L. Melvin, Principles of modern radar. Citeseer, 2010, vol. 1.

\bibitem{cvx}
M. Grant and S. Boyd, “CVX: Matlab software for disciplined convex programming, version 2.1,” http://cvxr.com/cvx, Mar. 2014.

\end{thebibliography}
\end{document}